 \documentclass[reqno,11pt,a4paper]{amsart}

\pdfoutput=1

\usepackage{amsfonts,amsmath,amssymb}
\usepackage[latin1]{inputenc}
\usepackage{dsfont}
\usepackage{enumerate}
\usepackage{xcolor}
\usepackage[all]{xy}
\def\notshow#1\notshowend{} %
\usepackage{hyphenat}
\usepackage{graphicx}

\usepackage{tikz}
\usepackage{tikz-3dplot}
\usepackage{pgfplots}

\usepackage{multirow}
\usepackage{array}

\usepackage{amssymb}
\usepackage{amsfonts}
\usepackage{mathrsfs}
\usepackage{hyperref}

\usepackage[normalem]{ulem} % PARA TACHAR

\def\br#1\er{{#1}} %
\def\bb#1\eb{\textcolor{blue}{#1}} 
\def\bm#1\em{\textcolor{magenta}{#1}} %

  \newcommand{\N}{\mathds{N}}     % Naturales
  \newcommand{\R}{\mathds{R}}     % Reales
   \newcommand{\CC}{\mathds{C}}     % Complejos
 
    \newcommand{\Lo}{\mathds{L}}     % Lorentz-Minkowski
\newcommand{\LL}{{\mathbb{L}}}

\newtheorem{thm}{Theorem}[section]
\newtheorem{prop}[thm]{Proposition}
\newtheorem{lemma}[thm]{Lemma}

\theoremstyle{definition}
\newtheorem{defi}[thm]{Definition}

\newtheorem{example}[thm]{Example}
\newtheorem{exe}[thm]{Example}

\newtheorem{rem}[thm]{Remark}

\newcommand{\be}{\begin{equation}}
\newcommand{\ee}{\end{equation}}
\newcommand{\ben}{\begin{enumerate}}
\newcommand{\een}{\end{enumerate}}
\newcommand{\bit}{\begin{itemize}}
\newcommand{\eit}{\end{itemize}}
\newcommand{\edoc}{\end{document}}

\newcommand{\bq}{\begin{quote}}
\newcommand{\eq}{\end{quote}}

\newcommand{\gs}{g_{\SSS^{n-1}}}

 \newcommand{\Hi}{\mathds{H}} % hiperblico
 \newcommand{\SSS}{\mathds{S}} % esfera

\title[]{A class of cosmological models with spatially constant sign-changing curvature}

\author[M. S\'anchez]{Miguel S\'anchez} \address{Departamento de Geometr\'{\i}a y Topolog\'{\i}a, %  F. Ciencias, 
\hfill\break\indent \& IMAG %Campus Fuentenueva s/n.
\hfill\break\indent Universidad de Granada, \hfill\break\indent 18071 Granada, Spain}\email{sanchezm@ugr.es}

\begin{document}

\begin{abstract} 
We construct globally hyperbolic spacetimes 
 such that each slice $\{t=t_0\}$ of the universal time $t$ is a model space of constant curvature $k(t_0)$ %(so,  intrinsically  isotropic and homogeneous) 
 which  
  may not only vary with $t_0\in\R$ but also change its sign. 
    The metric  is smooth and  slightly different to   FLRW spacetimes, namely, $g=-dt^2+dr^2+ S_{k(t)}^2(r) g_{\SSS^{n-1}}$, where $g_{\SSS^{n-1}}$ is the metric of the standard sphere, $S_{k(t)}(r)=\sin(\sqrt{k(t)}\, r)/\sqrt{k(t)}$ when $k(t)\geq 0$ and $S_{k(t)}(r)=\sinh(\sqrt{-k(t)}\, r)/\sqrt{-k(t)}$  when $k(t)\leq 0$.
  
   In the open case, the $t$-slices are (non-compact)	 Cauchy hypersurfaces of curvature $k(t)\leq 0$, thus homeomorphic to $\R^n$; a typical example is $k(t)=-t^2$ (i.e., $S_{k(t)}(r)=\sinh(tr)/t$). 
   In the closed case,   $k(t)>0$ somewhere,  a slight extension of the class  shows how the topology of the $t$-slices changes. 
   This  makes at least one  comoving  observer to disappear in finite time $t$ showing some similarities with an inflationary expansion. Anyway, the spacetime is foliated by Cauchy hypersurfaces  homeomorphic to spheres, not all of them  $t$-slices. 
\\

\noindent 
{\em MSC:}  83C15, 53C50,   58J45. \\

\noindent {\em Keywords:} cosmological models, space topology change, constant curvature, space isotropy and homogeneity, FLRW spacetimes, Cauchy foliation, flat space instant, inflation.  \\

\end{abstract}
\maketitle

%\newpage

%{\small
\tableofcontents
%}

\newpage

\section{Introduction}\label{s_Intro} 
%and motivation}

Friedman-Lema$\hat{\i}$tre-Robertson-Walker (FLRW) spacetimes constitute the standard class of cosmological spacetimes,  supported by hypotheses of isotropy %and homogeneity 
of the ``spatial'' part of the spacetime, which imply  that each spacelike slice $\{t=t_0\}$ of the universal time $t$ will have constant curvature $k(t_0)$. These hypotheses may be somewhat tricky because, depending on how they are formulated, they will imply  whether the sign of $k(t)$ must be constant or not  (see Remark \ref{r1}). 
% \bm (see Appendix~A??). \em  
 The aim of the present paper is to describe a simple class of  cosmological spacetimes whose metrics are  constructed from $g=-dt^2+dr^2+ S_{k(t)}^2(r) g_{\SSS^{n-1}}$
where  such a sign change occurs for the $t$-slices, the latter being the restspaces of the (freely falling) comoving observers at $\partial_t$. 
%just considering directly the natural expressions of the spatial metrics of constant curvature. 
%This change occurs globally, as required for a cosmological setting. Indeed,   
As far as the author knows, this specific  class is not   taken into account  in classical textbooks  on Relativity (such as 
\cite{BEE, Carroll, HE, MTW, SW, Wald, W}) nor %the research in 
in standard  Cosmology. 
 This is not the unique way to obtain a spatial curvature change, at least locally  (recall that different smooth families of spaces of constant curvature with non constant sign can be constructed  and its $t$-parametrization will make the job), and an independent systematic study is being carried out in   \cite{MV}.

Our  models will be globally hyperbolic and  diffeomorphic either to $\R\times \R^n$ ({\em open models}, when $k(t)\leq 0$ everywhere) or $\R\times \SSS^n$ ({\em closed models}, when  $k(t)>0$ somewhere).  
This  suggests   some possibilities in Cosmology and, so, its interest may go  beyond the  academic one. 
%\footnote{\bm For example,  independent experimental data (WMAP, BOOMERanG, Planck) confirm the space flatness of universe and could be confronted with our models \ref{00}.\em }. 
 Here, we focus on a rigurous  geometric exposition of the spacetimes, in order to show how the smooth curvature sign change occurs everywhere. In particular,  full technical details   prove that the spacetime metric is  smooth at the points where  spherical coordinates are not, that is, at the origin and \br its cut locus \er at each $t$-slice, the latter case being subtler and  necessary for the  closed model.  Here, we focus on a geometric description with no care on the stress-energy tensor. However, we emphasize that the additional assumptions to ensure differentiability in the closed model (including the somewhat more general expression of the metric in Def. \ref{d_BCTCCM}) may affect the stress-energy only a small region around a singular  observer, see Remark \ref{r_further}; thus, it might be acceptable in settings such as the inflationary one.

The  open models are very simple geometrically, because each slice $\{t=t_0\}$ becomes a Cauchy hypersurface of constant non-positive curvature $k(t_0)$. 
The  closed models are more involved, as  a topological change must occur in the (constant curvature) $t$-slices; in particular, not all the $t$-slices can be Cauchy. However, a Cauchy slicing by spacelike topological spheres (not all of them with constant curvature) can be found. 
This case is more involved, as one has to impose additional hypotheses to ensure smoothability at the spacelike cut locus. We will take a somewhat more general metric depending of a function $\varphi$ and make an illustrative specific choice  (Rem.~\ref{r_further}, Def. \ref{d_BCTCCM}) 
which will 
allow us to find the Cauchy slicing in a rather explicit way. Noticeably,  a {\em singular  comoving  observer} will emerge. This shows how  spacetime inhomogeneities  blowup in order to permit  the  smooth  topological change of the  %(non-Cauchy) 
$t$-slices,  in spite of   the  constancy of their curvatures. 
%and global hyperbolicity.  %\footnote{Part of this non-trivial construction (which leads to Prop. \ref{p_sco}) can be circumvented by using  general results on Cauchy slicings, see  Rem. \ref{r_propag0}.}.

When $n=3$, our local models lie
  in  a subclass of the
 Stephani Universes \cite{St}, which  was rediscovered 
 by Krasi\'nski \cite{Kr1, Kr2}   by assuming  $O(3)$ symmetry on the spatial slices. It was also  studied systematically 
 %from a general   %local and global 
% viewpoint 
by Sussman \cite{Su0, Su}. The   possibility of  a  change of sign for certain spacelike foliations appears in some of these references (see the detailed exposition in \cite[Chapter 4]{Kr3}, especially \S~4.10 and the  historical note at p.~148), even though in a less restrictive sense than above, see our  discussion at the end of \S~\ref{s43}. 

This paper is organized as follows. 

In Section \ref{s2} we start giving a technical unified expression for the Riemannian model spaces of constant curvature $k$. Such an expression  depends on some functions $S_k(r)$ (involving sines and hyperbolic sines) which are shown to be smooth also in the parameter $k$ (Lem. \ref{l1}, Rem. \ref{r_derivadas_regla_cadena_fallida}). This technical result allows us to obtain local variations of the spatial curvature $k(t)$, including its sign  (Theorem \ref{t_local}). %As far as we know, this local case had neither  been taken into account  before (Rem. \ref{r1}). 

In Section \ref{s_model_open}, we show that the local model can be extended globally to provide a change between flat  and negative  curvature ($k(t)\leq 0$), giving rise to the open models  (Theorem~\ref{t_open}). A technical question for this global extension is to check that the metric expressions in spherical coordinates are smooth even when $r=0$, consistently with the $O(n)$ invariance of the model. Indeed, it is worth pointing out that a priviledged {\em centered  comoving  observer}  appears at $r=0$ (Def.~\ref{d_2ndff_slices_open}). This observer shows explicitily the existence of spacetime anisotropies between  comoving  observers, in spite of the intrinsic  isotropy of their restspaces (i.e., the $t$-slices). In fact, the anisotropies are apparent when the second fundamental form of the slices are computed  (Prop. \ref{p_2ndff_slices_open}). The  simple  case $k(t)=-t^2$ is considered explicitly (Ex. \ref{e_sinh}) and, in general, the simplicity of all these open models might make them  useful for  several  purposes.

In Section \ref{s_model_closed}  we consider the  closed models, when $k(t)>0$ somewhere. Being more involved, this case is developed in three steps.  In \S \ref{s41}  the toy model  with slices of dimension $n=1$ is considered. Of course, such slices are necessarily flat, but one can still model a topological transition from the circle $\SSS^1$ to the line $\R$, so that the globally hyperbolic spacetime matches  $(-\infty,0)\times \SSS^1$ with $[0,\infty)\times \R$.  As a working definition to give an illustrative idea,  this is called  a {\em basic cosmological topological change model} (BCTCM), see Defn.~\ref{d_BCTCM},  Prop.~\ref{p_dim2_a}. From the technical viewpoint, we introduce a function $\varphi(t,\theta)$ which will permit both, the smooth extension to the whole cylinder when $t<0$ and the smooth matching when $t\geq 0$.
%; its necessity   
%is discussed in the Appendix. 
%(where the different viewpoint with  \cite{Kr1, Kr2, Su} is also emphasized). 
Then, a noticeable object emerges  under these choices, 
 namely, the {\em singular  comoving  observer} $\gamma_\pi$ at $(-\infty,0)\times\{e^{i\pi}\}$ 
(in addition to the previous centered  comoving  observer). 
   That observer is inextendible to $t=0$, anyway, all the Cauchy hypersurfaces must intersect it. Indeed, an explicit foliation  by Cauchy hypersurfaces of the spacetime can be found using $\gamma_\pi$ (Prop.~\ref{p_sco}, Rem. \ref{r_propag0}). In \S \ref{s42}, we  consider higher spatial dimensions $n\geq 2$ and show that the basic causal properties of  our example for   $n=1$ still hold; in particular, both the centered 
   and singular \br comoving \er observers appear consistently with the $O(n)$ invariance of the model. However,  now the topological change is also a true curvature change  from $k(t)>0$ when $t<0$ to $k=0$ when $t\geq 0$ (this is called a basic cosmological  topological and  curvature change model, BCTCCM in Defn.~\ref{d_BCTCCM}) Theorem~\ref{t_closed}). Finally, in \S \ref{s43}, we show that such a change can be easily adapted to more general situations, so that, in particular, one can construct  transitions from $k(t)>0$ to $k(t)<0$.
 Some conclusions are given in the last section.

%\newpage

\section{Local change of the spatial curvature sign}\label{s2}

\subsection{Unified expression of the Riemannian model spaces $M_k$}

Consider the functions 
$$S_k(r):=
\left\{ \begin{array}{lll}
\frac{\sin(\sqrt{k}\, r)}{\sqrt{k}} & & \hbox{if} \; k>0 \\

r & & \hbox{if} \; k=0 \\

\frac{\sinh(\sqrt{-k}\, r)}{\sqrt{-k}} & & \hbox{if} \; k<0 
\end{array}\right. \qquad  
C_k(r):=
\left\{ \begin{array}{lll}
\cos(\sqrt{k}\, r) & & \hbox{if} \; k>0 \\

1 & & \hbox{if} \; k=0 \\

\cosh(\sqrt{-k}\, r) & & \hbox{if} \; k<0, 
\end{array}\right. 
$$
 $ r\in \R$, with derivatives  $S'_k=C_k$, $C'_k=-k S_k$ and $C_k^2+k S_k^2=1$. 
 
 %and put 

%$$\begin{array}{c} T_k: \left(-\frac{\pi}{2\sqrt{k}}, \frac{\pi}{2\sqrt{k}}\right)\rightarrow \left(-\frac{1}{\sqrt{-k}}, \frac{1}{\sqrt{-k}}\right), \quad r \mapsto T_k(r):=\frac{S_k(r)}{C_k (r)}, \\ 
%\hbox{(convention $1/\sqrt{k}=\infty$ if $k\leq 0$)}. \end{array}$$
 % $T_k$ is a diffeomorphism with $T_k'=1+kT_k^2$ satisfying
%$
%T_k(r-r')=(T_k(r)-T_k(r'))$ $/(1+k \, T_k(r)\, T_k(r'))$.

Let $M_k$ be the $n$-Riemannian model space of curvature $k$ ($n\geq 2)$, i.e., 
$$ M_k(r):=
\left\{ \begin{array}{lcl}
\SSS^n(r_k) & \hbox{
($n$-sphere of curvature $k=\frac{1}{r_k^2}$)}
 & \hbox{if} \; k>0 \\
\R^n & \hbox{($n$-Euclidean space)} & \hbox{if} \; k=0. \\

\Hi^n(r_k) & \hbox{
($n$-hyperbolic space of curvature $k=\frac{-1}{r_k^2}$)} & \hbox{if} \; k<0 
\end{array}
\right.
$$
The metric $g_k$  of $M_k$ can be written using normal spherical coordinates as\footnote{\br This type of expressions for functions $S_k, C_k$  with a prescribed $k$ are well-known in Riemannian comparison theory, see for example \cite[\S 3.1]{Chavel} or \cite{AG}. \er }:
%the metric $g_k$ of $M_k$:
\begin{equation}\label{e_spherical}
 g_k= dr^2 +S_k^2(r) \, \gs, \qquad r\in (0,d_k), \; \hbox{with} \;
d_k:= \pi/\sqrt{k}
\end{equation}
(under the convention $1/\sqrt{k}=\infty$ if $k\leq 0$), where $g_{\SSS^{n-1}}$ is the metric of the standard  unit $(n-1)$-sphere.
Recall that  $g_k$  is always smoothly extensible to $r=0$. When $k>0$, $g_k$ is also extensible to $d_k=\pi/\sqrt{k}$ on a topological sphere,  
so that
 $M_k$ can be seen extrinsically as a sphere of radius  $r_k:=1/\sqrt{k}$ in $\R^{n+1}$; intrinsically, however, the  diameter of this sphere (i.e. the supremum of the distance between each two points) is $d_k= \pi/\sqrt{k}$.
When $k\leq 0$,  the expression \eqref{e_spherical}  is  defined for  $r\in (0,\infty)$ (the intrinsic diameter $d_k$ is infinity) and $g_k$ becomes a metric on the whole $\R^n$.

\subsection{Smoothness of the variation with $k$ and local transition model}
For $k=\epsilon=1,0,-1$, one has the elementary Maclaurin series: 
\be \label{e_Sgeneral0}
S_\epsilon(r)= \sum_{m=0}^\infty (-\epsilon)^m \frac{r^{2m+1}}{(2m+1)!}=r - \epsilon \frac{r^3}{3!}+  \epsilon^2 \frac{r^5}{5!}- \epsilon \frac{r^7}{7!}+...
\ee
where $(\epsilon=0)^0:=1$. This can be used to prove the smoothness of $S_k(r)$ with $k$ and, then, to construct a local transition model of curvature.

\begin{lemma}\label{l1} The function
$
\R^2\ni (k,r) \mapsto S_k(r)
$
is analytic. Thus, for any smooth function
  $\R\ni t \mapsto k(t)$, it is also smooth 
\be \label{e1}
\R^2\ni (t,r) \mapsto S(t,r):= S_{k(t)}(r).
\ee
\end{lemma}
\begin{proof}
Let $\epsilon(k)$ be the sign of $k$, using \eqref{e_Sgeneral0} (with  $-\epsilon(k) |k| =-k$):
\be \label{e_Sgeneral}
\begin{array}{rl}
S_k(r)= & \br \frac{1}{\sqrt{|k|}} S_{\epsilon(k)} (\sqrt{|k|} \; r) = \er \sum_{m=0}^\infty (-\epsilon(k))^m \frac{(\sqrt{|k|}\, r)^{2m+1}}{\sqrt{|k|} (2m+1)!} %=
%r  -\epsilon(k) |k| \frac{r^3}{3!}+  
%%%\epsilon(k)^2 
%|k|^2 \frac{r^5}{5!}- \epsilon(k) |k|^3 \frac{r^7}{7!}+... \\
%=& r  - k \frac{r^3}{3!}+  %\epsilon(k)^2 
%k^2 \frac{r^5}{5!}- k^3 \frac{r^7}{7!}+...
\\
= &  \sum_{m=0}^\infty (-\epsilon(k))^m |k|^m \frac{r^{2m+1}}{(2m+1)!}
=
\sum_{m=0}^\infty (-k)^m \frac{r^{2m+1}}{(2m+1)!},
\end{array}
\ee
%and recall that the last  series has infinity radius of convergence.
which is \br norm convergent for  $|k|<1$ and, then, analytic everywhere (see for example \cite[Lemmas A1, A2]{Ross}). \er
\end{proof}

\begin{rem}\label{r_derivadas_regla_cadena_fallida}
(1) Notice that the functions above are smooth even if $|k|$ or $\sqrt{|k|}$ (eventually regarded as functions of $t$) are not smooth at 0. Indeed, the derivatives can be obtained  by derivating directly the terms in the series \eqref{e_Sgeneral}. 
In particular, for $k$ smooth in $t$ with derivative $k'$,
\begin{equation} \label{e_ejemplo_kprimat0}
\partial_t S(t,r)= k'(t) \,  \partial_k(S_k(r))_{k=k(t)},
\quad 
\partial_k (S_k(r))=
-\sum_{m=1}^\infty m (-k)^{m-1} \frac{r^{2m+1}}{(2m+1)!}.
\end{equation}
Thus, $\partial_t S(t,r)=0$ whenever $k'(t)=0$ and  $\partial_k (S_k(r))|_{k=0}=-r^{3}/6$.

(2) The previous observation should be taken into account even when $k\leq 0$, as in the open models below. Indeed, taking into account \eqref{e_Sgeneral},  define 
\begin{equation}\label{e_f}  
f(z)= \sum_{m=0}^\infty  \frac{z^{2m}}{(2m+1)!}, \forall z\in \R,  \quad \hbox{so that} \quad
S_{k}(r)=rf\left(\sqrt{-k}\; r\right), 
\end{equation}
for all $k\leq 0$ and $r\in\R$. Even though  $f$ is smooth  the chain rule should not be used in \eqref{e_f}. In fact,  take a function  $\R\ni t\mapsto k(t)\leq 0$ and put
$$S(t,r) (=S_{k(t)}(r))= rf\left(\sqrt{-k(t)}\; r\right).$$
%As $k(t)\leq 0$, 
Then, $t \mapsto -k(t)=|k(t)|$ may be smooth 
%whenever its first non-vanishing derivative is even. However, even in this case 
but $\sqrt{-k(t)}$ may be non-smooth (say, $k(t)=-t^2$, $\sqrt{-k(t)}=|t|$). 
% Anyway, a formal chain rule in \eqref{e_f}  may work taking into account that jump discontinuities as the previous ones would be compensated by $f'(0)=0$. 
 
 (3) A shortcut to avoid such subtleties later would be to choose   $k(t)$ such that, whenever $k(t_0)=0$, then all the   derivatives of $k$ vanish at $t_0$. However, this would not be a big simplification and would exclude simple  choices as above (used in  Ex. \ref{e_sinh} below).
\end{rem}

%\subsection{Local transition model}\label{s_model}
Now, let us construct a spacetime in an open set of $\R^{n+1}$ with a change of the sign of the spatial curvature
 for any choice of $k(t)$ . 

\begin{thm} \label{t_local}
Let $k: \R \rightarrow \R$ be any smooth function and consider the open subset $U \subset \R\times\R^n$ defined taking spherical coordinates in $\R^n$ as
$$
U=\left\{(t,x): 0<r(x)<d_{k(t)}:=\pi/\sqrt{k(t)}\right\} \qquad (d_{k(t)}:=\infty \; \hbox{if} \; k(t)\leq 0)$$
endowed with  %the metric
$g=-dt^2+dr^2+ S_{k(t)}^2(r) g_{\SSS^{n-1}}$. 

Then  $g$  is a smooth Lorentzian metric on $U$ and each  slice $\{t=t_0\}$ has constant curvature $k(t_0)$.
\end{thm}

\begin{proof}
Given the function $k(t)$,  $U$ is chosen so that the spherical coordinates are smoothly well defined in each slice $\{t=t_0\}$ and the continuity of $d_k$ with $k$ as a function on $(0,\infty]$ implies that $U$ is open.  Thus, the smoothness of the tensor $g$ is just a consequence of the smoothness of $S(t,r)=S_{k(t)}(r)$ in \eqref{e1}, ensured by Lemma \ref{l1}. As $S_{k(t)}^2(r)>0$ on $U$, $g$ is Lorentzian, and the curvature $k(t_0)$ of each slice follows from \eqref{e_spherical}.  
\end{proof}

\begin{rem}\label{r1} (1) Being each slice $\{t=t_0\}$ of constant curvature, it is   locally isotropic, that is, each $p$ in the slice admits some neigborhood $W\subset \{t=t_0\}$ such that for  each two tangent directions $v,w$ of the slice at $p$ there exists a isometry of $W$ which maps $v$ in $w$.\footnote{The slice is also locally homogeneous (i.e. any two points $p, q$ in $\{t=t_0\}$, admit  neighborhoods $W_p, W_q, \subset \{t=t_0\}$  and a local isometry of the slice which maps $p$ into $q$), which is also a usual assumption for convenient spatial slices.} 
 However,    
a local isometry of the {\em  spacetime} mapping the direction $v$ into $w$ and preserving the slices will not exist  in general. Indeed,  such a property would forbid the change of sign for $k(t)$, see \cite[p. 342, Prop. 6]{O} (compare with \cite[\S 5.1]{W})\footnote{Compare  also with  the intrinsic characterization of Generalized Robertson-Walker spacetimes  in \cite{Sa98} 
(the metric of these spacetimes  was introduced more directly  in \cite{ARS}).}. 
This stronger condition has not been always taken into account in the standard literature  (see for example \cite[p. 112-113]{C-B}) and, thus, our metrics in Theorem \ref{t_local} might have been  considered  %taken into account 
then; a detailed study is  carried out in~\cite{Av}.  

(2) Our metrics above lie in a particular case of  spherical symmetry; indeed,  $S(t,r)$ corresponds with the standard function $Y(t,r)$  in the book by Stephani et al. \cite{St_et_al}, formula (15.9). In the regions where the gradient of $Y$ is timelike or spacelike, sometimes it is used  as a canonical ``$t$'' 
or ``$r$'' coordinate; however,    the gradient of our $S(t,r)$ may be lightlike (recall~\eqref{e_ejemplo_kprimat0}). We emphasize that, when working with such a  $Y(t,r)$, the study is not global (even if $Y(t,r)$ is smooth everywhere). In fact, one has to check whether $g$ matches smoothly with the spherical coordinates at $r=0$ or $r=d_{k(t)}$, which will be a posteriori the boundary of $U$ in the whole spacetime. Typically, the case $r=0$ comes from a direct computation (as in our open model below); however, the case  $r=d_{k(t)}<\infty$ must be carefully taken into account, as we will do in the closed case. 

(3)  The metric $g$ can be regarded as a warped product with base \br $(U_B,-dt^2+dr^2)$, where $U_B=\{(t,r)\in \R^2: 0<r<d_{k(t)}\}$, \er fiber the sphere $\SSS^{n-1}$
%, g_{\SSS^{n-1}})$ 
and warping function $S(t,r)=S_{k(t)}(r)$; so, O'Neill's formulas  for geodesics and curvature  \cite[Ch. 7]{O} apply. In particular, the leaves $U\times \{u_0\}$, $u_0\in \SSS^{n-1}$, are totally geodesic and the integral curves of $\partial_t$ are geodesics in the whole spacetime, that is, the comoving observers are freely falling. \br The Ricci curvature and, then, the Einstein tensor with arbitrary cosmological constant, 
can also be  computed 
by using the warped structure (see the Appendix). However, their properties depend strongly on the choice  of $k(t)$ and  physical applicability would be analized elsewhere. \er
\end{rem}

\section{Open cosmological spacetimes}\label{s_model_open} 
Next, let us  go from the local to the global model for curvature sign change  starting at the open model, $d_{k(t)}=\infty$ for all $k(t)$. This corresponds to the continuous extension of the metric $g$ in Theorem \ref{t_local} from $U=\R^{n+1}\setminus \{r=0\}$ (i.e., $\R^{n+1}\setminus (\R\times \{0\})$) 
to the whole $\R^{n+1}$. We will check that this extension is also smooth as well as  other announced properties.

\begin{thm}\label{t_open} For any smooth nonpositive function  $\R \ni t \mapsto k(t)\leq 0$,  
\be \label{e_metrica_gen}
g=-dt^2+dr^2+ S_{k(t)}^2(r) g_{\SSS^{n-1}} 
%\qquad \hbox{on} \; \R\times \R^n
\ee
is a smooth Lorentzian metric on the whole $\R^{n+1}=\R\times \R^n$ with slices $\{t=t_0\}$  isometric to the model space of curvature $k(t_0)\leq 0$ (an Euclidean or hyperbolic space).
Moreover, each slice $\{t=t_0\}$ is a Cauchy hypersurface.

\end{thm}
\begin{proof}
   Theorem \ref{t_local} reduces the first assertion to  prove  that the spherical expression becomes  smooth at $r=0$. With this aim, we will rewrite the metric using cartesian coordinates $x^i$  in $\R^n$. Indeed, using the function $S(t,r)$ in \eqref{e1}, it is enough to prove the smoothness at $r=0$ of $g+dt^2$ in \eqref{e_metrica_gen}, that is, 
\begin{equation}
\label{e_cartesian}
\begin{array}{rl}
 dr^2 +\frac{S^2(t,r)}{r^2}r^2 \gs
  & = dr^2 +\frac{S^2(t,r)}{r^2}(\sum_{i=1}^n (dx^i)^2-dr^2) \\
  &
= \frac{S^2(t,r)}{r^2} \sum_{i=1}^n (dx^i)^2+
 \left(1- \frac{S^2(t,r)}{r^2}\right) dr^2.
\end{array}
\end{equation}
Let us analyze the last two terms. For the first one, $S(t,r)/r$ is the composition of the smooth function 
$(t,r)\mapsto (k(t),r)$ and the function $S_k(r)/r$. Using  
\eqref{e_Sgeneral} for the latter
\begin{equation}\label{e_cartesian1}
\frac{S_k(r)}{r}=
\sum_{m=0}^\infty (-k)^m \frac{r^{2m}}{(2m+1)!}=
\sum_{m=0}^\infty (-k)^m \frac{(\sum_{i=1}^n(x^i)^2)^{m}}{(2m+1)!},
\end{equation}
which is analytic for all $(k, x^i)\in \R\times \R^n$.
For the last term in \eqref{e_cartesian}, as $r^2 dr^2$ is always smooth, the problem is reduced to  the smoothness of
$$\frac{1}{r^2}\left(1- \frac{S^2(t,r)}{r^2}\right)=
\frac{1}{r^2}\left(1- \frac{S^2_{k(t)}(r)}{r^2}\right)=
\frac{1}{r^2}\left(1- \frac{S_{k(t)}(r)}{r}\right)\left(1+ \frac{S_{k(t)}(r)}{r}\right).
$$
Reasoning as above, the result follows from the analyticity of
\begin{equation}\label{e_cartesian2}
\begin{array}{rl}
\frac{1}{r^2}\left(1- \frac{S_{k}(r)}{r}\right)= &
- \frac{1}{r^2}\sum_{m=1}^\infty (-k)^m \frac{(\sum_{i=1}^n(x^i)^2)^{m}}{(2m+1)!} \\
= & 
%- \sum_{m=1}^\infty (-k)^{m} \frac{(\sum_{i=1}^n(x^i)^2)^{m-1}}{(2m+1)!}= 
- \sum_{m=0}^\infty (-k)^{m+1} \frac{(\sum_{i=1}^n(x^i)^2)^{m}}{(2m+3)!}.
\end{array}
\end{equation}

The Cauchy character of the slices follows because 
$S_{k(t)}(r)\geq r$ (use \eqref{e_Sgeneral} noticing $-k\geq 0$) and, thus, the cones of $g$ are narrower than those of $\Lo^{n+1}$ (see \cite[Prop. 3.1]{S_Penrose} for background and a more general result).
\end{proof}

Each slice $\{t=t_0\}$ is isometric to a model space and, thus, it is globally   isotropic and  homogeneous (recall Rem. \ref{r1}). However, the way how the slice is embedded in the spacetime does not satisfy these properties, as checked next.

\begin{prop} \label{p_2ndff_slices_open} The second fundamental form $II_{t_0}$ (with respect to the future direction given by $\partial_t$) of the slice $\{t=t_0\}$ satisfies:
\begin{equation}\label{e_2ndff_slices_open}
II_{t_0}=  S(t_0,r)\partial_tS(t_0,r)\gs
%r^3 \cdot (\sqrt{-k})'(t_0) \cdot \left[(f\cdot f')(\sqrt{-k(t_0)} \; r)\right] \, \cdot \gs ,
\end{equation}
In particular: (a) if the slice  $\{t=t_0\}$ of the open model is flat then it is totally geodesic, and (b)  $II_{t_0}$ vanishes on the point $r=0$.
\end{prop}

\begin{proof} As $\partial_t$ is unit and normal to the slices, the expression follows from $II_{t_0}= \frac{1}{2} \partial_tg$
%= S(t_0,r)\partial_tS(t_0,r)\gs .$$
The assertions (a) and (b) follow from the expression of these derivatives in Rem. \ref{r_derivadas_regla_cadena_fallida} (for (a), recall that necessarily $k'(t_0)=0$).
\end{proof}
Notice that the integral curves of $\partial_t$ can be regarded as  comoving  observers, and $r=0$ determines one of them, namely, $\R\ni t\mapsto (t,r=0)$). Taking into account the assertion (b),  this observer is priviledged for the intrinsic geometry of the spacetime, up to trivial case $k\equiv 0$ (i.e., $\LL^{n+1}$). 
The following definition allows us to summarize the introduced notions and results.
\begin{defi}\label{d_2ndff_slices_open}  The spacetime
$(\R \times \R^n, g)$ defined in Theorem \ref{t_open} is the {\em open cosmological model} of  spatial curvature 
function\footnote{With more generality, one can consider (here and later, in the case of closed models) the possibility $I\times \R^n$, 
with $I\subset \R$ an open interval. The simplification $I=\R$ is essentially notational and enough for the properties to be considered here. 
%no more generality would be obtained, as one could re-scale 
%the function $k$ with any diffeomorphism 
%$I \rightarrow \R$.
}  $t\mapsto k(t)$.
 The integral curve of $\partial_t$
at $r=0$ will be called the {\em  centered  comoving  observer}.
\end{defi}

%\begin{rem}\label{r_2ndff_slices_open} 
Even though the geometry of the spacetime priviledges  the centered  comoving  observer  (except in $\LL^{n+1}$), we  emphasize  that it cannot be determined by using only the intrinsic geometry of each slice (indeed, we have used the extrinsic one). 
%\end{rem}

\begin{exe} \label{e_sinh} The simple  choice   $k(t)=-t^2$,   that is\footnote{More precisely, $S_{k(t)}(r)=\frac{\sinh (|t| r)}{|t|}=\frac{\sinh (t r)}{t}$.}, 
\be \label{e_metrica_part}
g=-dt^2+dr^2+  \frac{\sinh^2(t r)}{t^2}  g_{\SSS^{n-1}},
\ee
has slices $\{t=t_0\}$ with intrinsic curvature $-t_0^2$. The slice $\{t=0\}$ is flat and totally geodesic. The comoving  observers at $\partial_t$ will find  a bouncing,  that is, a
contraction (they approach each other) before this slice and an expansion after it\footnote{\label{f_sinh} With obvious modifications,  these properties are shared by  the flat slices $\{t=t_0\}$ 
of  all the open cosmological models, whenever $k$ is not constant in $[t_0-\delta, t_0]$ or $[t_0, t_0+\delta]$ for some $\delta>0$.}. Using \eqref{e_2ndff_slices_open}, 
the other slices have the  second fundamental form: 

$$
II_{t_0}= 
 \;  \frac{\sinh(t_0 r)}{t_0^2} \;
\left( r\cosh(t_0 \, r)-\frac{\sinh(t_0 \, r)}{t_0}\right)
  \gs .
$$
Clearly, this expression is anisotropic, as $II_{t_0}$ vanishes in the radial directions
(in particular,  along the centered comoving  observer $r=0$) but grows  with $r$ in the directions orthogonal to the radial ones.
\end{exe}

\section{Closed cosmological spacetimes}\label{s_model_closed}
The case when $k>0$ somewhere becomes subtler, as $M_k$ is a sphere with intrinsic diameter  $d_k= \pi/\sqrt{k}$. From the local viewpoint, we handled this case with no special caution. However, from the global one, if $k$ changes from positive to 0 (and eventually negative) then a topological change will occur. Thus, in order to highlight the relevant global ideas, we will start by focusing on the case of a smooth function $k$ satisfying:
\begin{equation}\label{e_curvat_top_change}
k(t) \left\{
\begin{array}{ll}
> 0 & \hbox{if} \; t<0 \\
= 0 & \hbox{if} \; t\geq 0.
\end{array}
\right.
\end{equation}
%As explained in Remark \ref{r_signok}, the function $S(t,r)$ will be smooth  in this case and, moreover, the order of differentiability could be tracked easily if the assumption $k(t)=0$ for $t>0$ was dropped. 

The strictly topological subtleties of the change appear clearly in the toy model  $n=1$ (so that the curvature $k(t)$ of all the slices is necessarily 0), which will be studied first in \S \ref{s41}. Then we will check that this extends to a true curvature change from positive to 0 curvature when  $n\geq 2$ in \S \ref{s42} and, finally, arbitrary curvature changes will be achieved in \S \ref{s43}.

\subsection{Spatial dimension $n=1$}\label{s41} In this case, the curvature of the slices is not taken into account but the extrinsic radius $R$ will play the role  of 
$1/\sqrt{|k|}$ in  higher dimensions. In order to achieve the topological change, we will have to make an involved construction by distinguishing a pair of comoving  observers. 
%again a kind of centered comoving observer as well as a second one. 
%The Appendix B   (which can be skipped) 
%might illustrate the necessity of such a construction. 

We will consider $\theta$ as the natural coordinate of $]-\pi,\pi[$; this interval, eventually, will be identified with a circle but a point, namely $\SSS^1\backslash\{e^{i\pi}\} (\subset \CC)$. Next, consider the following metric 
\begin{equation}
\label{e_closed metric}
g_B=-dt^2+\varphi^2(t,\theta) d\theta^2 \qquad (t,\theta)\in \R\times ]-\pi,\pi[
\end{equation}
where $\varphi$ will be a suitable  function such that $g_B$ is extensible to $(t,e^{i\pi})$ for $t<0$  but not for $t=0$.  
Specifically, $\varphi$ is any smooth function satisfying:
\ben
%\begin{array}{l}
\item[(a)] \label{1} $
\varphi(t,\theta)\geq 1$ and non-decreasing with $t\in\R$, for each $\theta\in ]-\pi,\pi[$. 
\item[(b)] \label{2}
  for each $m\in \N$ (positive integer)
\begin{equation}\label{e_varphi}
\left\{
\begin{array}{llll}
\varphi(t,\theta)\geq  m e^m  ,  & \hbox{if} & \; t\geq -\frac{1}{m}, & \theta\not\in [-\pi+\frac{1}{m}, \pi-\frac{1}{m}] \\
\hbox{$t\mapsto \varphi(t,\theta)$ is constant}, & \hbox{if} & \; t\geq -\frac{1}{2m}, &  \theta\in [-\pi+\frac{1}{m}, \pi-\frac{1}{m}]
\end{array}
\right.
%\end{array}
\end{equation}
\item[(c)] \label{3} $\varphi(t,\theta)=\varphi(t,-\theta), \forall \theta\in ]-\pi,\pi[$  and $\varphi$ can be smoothly extended to $(t,e^{i\pi})$ when $t<0$. Moreover\footnote{\label{foot_extremoPi}  These two additional  conditions here will be used only in the case $n\geq 2$ later. Notice, however, that 
both of them will be satisfied by the explicit $\varphi$ in Example \ref{ex1} (in particular, the auxiliary $\hat \varphi(t,\cdot)$ therein will be equal to 1 in $(-\pi+1,\pi-1)$). The second condition implies the previously required smooth extendability of $\varphi$ to $(t,e^{i\pi}), t<0$ (in fact, $\varphi(t,\theta )$ depends only on $t$ close to $(t,e^{i\pi}), t<0$).  Of course, this independence of $\theta$ can be weakened by assuming other more accurate conditions, such as the vanishing of enough partial derivatives for $\theta$ on  $(t,e^{i\pi})$.},
 $\varphi(t,\cdot)$ is  locally equal to 1 around $\theta=0$ and, for each $t<0$, $\varphi(t,\cdot)$ is locally constant in a small neighborhood of $e^{i\pi}$ of radius $\epsilon(t) \searrow 0$. 
\een
Notice that the first condition in \eqref{e_varphi} yields:
\begin{equation}\label{e_varphi_a}
\begin{array}{lll}
\ell(t):=\int_{-\pi}^{\pi}\varphi(t,\theta)d\theta >  2 e^m,  &
\hbox{if $t\geq -\frac{1}{m}$ }\;
 (\hbox{thus,} \; \, \ell(t)=\infty, \forall t\geq 0).
\end{array}
\end{equation}

The second condition in \eqref{e_varphi} ensures the existence of a smooth function 
$\varphi_0:]-\pi,\pi[ \rightarrow [1, \infty )$ such that 
\begin{equation}\label{e_varphi_b}
\varphi(t,\theta)=\varphi_0(\theta)\geq 1 \qquad  \forall t\geq 0, 
\end{equation}
 indeed, this is valid also in a neighborhood of the closed half strip $t\geq 0$.  Using again the first condition (recall \eqref{e_varphi_a}),
\be \label{e_varphi_b_bis}
\int_{-\pi}^0 \varphi_0(\theta)
d\theta= \int^{\pi}_0 \varphi_0(\theta)d\theta=\infty.
\ee

\begin{figure}
	\centering
\includegraphics[height=0.4\textheight]{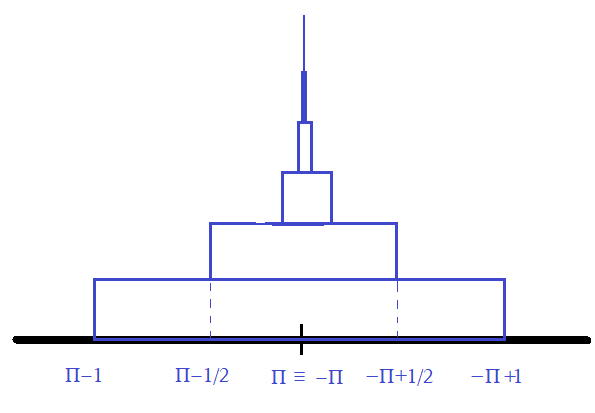}
\caption{\label{fig1} Function 
$\hat\varphi(t=0,\theta)-1$
in $(\pi-1,\pi)$ and in $ (-\pi,-\pi +1)$ (it would have infinite ``steps'' do that it diverges in 
  $\theta=\pm \pi$).  
    For $t=-1$, only the first step would appear, and a new $m$-th step would be included at each  $t=-1/m$, $m\in \N$.  
}\end{figure}

\begin{example}\label{ex1} {\em (Explicit $\varphi$)}. In order to give a construction of  $\varphi$,  start with: %the function
$$
\hat \varphi(t,\theta):=1+\sum_{k=1}^{\hbox{\tiny{m=Int}}[-1/t]}  e^k  \chi_{]-\pi, 
-\pi+\frac{1}{k}]\cup [\pi-\frac{1}{k},\pi[}(\theta) \qquad \forall t< 0, \; \; \theta \in ]-\pi,\pi[,
$$
where Int$[-1/t]$ denotes the integer part of $-1/t$ and $\chi$ is the characteristic function of the corresponding set (equal to 1 on  $]-\pi, 
-\pi+\frac{1}{k}]\cup [\pi-\frac{1}{k},\pi[$ and 0 otherwise), see Fig. \ref{fig1}. This function satisfies all the requirements but smoothness in its domain. In particular, $\hat \varphi$ can be continuously  extended to $(t,e^{i\pi})$ when $t<0$ and to the region $t\geq 0$ as in \eqref{e_varphi_b}. 

So, $\varphi$ can be chosen by smoothing $\hat \varphi$ ensuring: (a) it is non- decreasing with $t$, (b) it satisfies, say, $\hat \varphi \leq \varphi \leq  \hat \varphi +1$ and (c) it remains invariant under $\theta \mapsto -\theta$. These properties can be ensured by using a natural bump function to smooth each new summand, whenever\footnote{Recall also that 
$\hat\varphi$ can be easily approximated by a continuous function 
satisfying (a), (b) and (c) (replace the vertical segments in its 
graph by slightly inclined ones which are chosen symmetrically with respect to $e^{i\pi}$). Then, one can use general results of 
approximation of continuous functions by smooth ones (see  \cite[\S 2.3.5]{S_Penrose} for background, even in the analytic case).}  $t=-1/m$.  
\end{example}

\begin{rem}\label{r_marc}
  It is worth emphasizing that the conditions imposed for $\varphi$ are sufficient for our purposes but they are not optimized (recall \br the discussion in \er footnote~\ref{foot_extremoPi}). Indeed,  the different approach in \cite{MV} implies that, \br for any function $\ell(t)$ which is smooth and satisfies  $0<\ell(t)<\infty$ for $t<0$ and joins continuosly  $\ell(t)=\infty$ for $t\geq 0$, \er the choice
 $$
 \varphi(t,\theta)= \frac{1+ %\frac{1}{4}
 \tan^2(\theta/2)}{1+\frac{4\pi^2}{\ell(t)^2}\tan^2(\theta/2)}
 $$  
will  also work\footnote{The author acknowledges Marc Mars for  suggesting  this idea. }, even if it does not fit exactly under our conditions. Although such an optimization is possible, our hypotheses may be enough to understand the qualitative behaviour of the model or to make rough estimates on cases such as the inflationary one.

\end{rem} 
 
 \begin{defi}\label{d_BCTCM} A {\em basic cosmological topological change model (BCTCM)}  is the manifold
 $$
 B=\left(\R\times \SSS^1\right) \setminus \{(t, e^{i\pi}): t\geq 0\}
 $$ 
 endowed with a metric $g_B$ as in \eqref{e_closed metric}  with $\varphi$ satisfying the hypotheses (a), (b) and (c)  therein, and continuously (then, smoothly) extended to  $e^{i\pi}$ when $t<0$.
 \end{defi}
 Recapitulating, we have:
 \ben
 \item From the topological viewpoint, $B$ is a cylinder, i.e., homeomorphic to $\R\times \SSS^1$.
 \item The metric $g_B$ is smooth on the whole $B$. 

\item Each slice $\{t=t_0\}$ is isometric to $\R$ when $t_0\geq 0$ and  to a cylinder of length $\ell (t_0)<\infty$ when $t_0<0$, being $\lim_{t_0\rightarrow 0}  \ell (t_0)=\infty$.

\item  For $t\geq 0$ one has (recall \eqref{e_varphi_b} and \eqref{e_varphi_b_bis}): 
\begin{equation} \label{e_et_tmayor0}
g=-dt^2+\varphi^2_0(\theta)d\theta^2=-dt^2+dx^2, \qquad x\in \R,
\end{equation}
the latter taking as a new coordinate $x(\theta)= \int_0^\theta \varphi_0(\bar \theta)d\bar \theta, \forall \theta \in ]-\pi, \pi[$.
 \een
%The following global properties become relevant.

\begin{prop}\label{p_dim2_a}
Any BCTCM is globally hyperbolic, with Cauchy hypersurfaces homeomorphic to $\SSS^1$.
The slices $t=t_0$ are Cauchy for $\{t_0<0\}$ (and they are not for $t\geq 0$).
\end{prop}
\begin{proof} As $t$ is a time function, global hyperbolicity follows by proving that $J^+(p)\cap J^-(q)$ is compact for any $p,q\in B$. This is trivial if $p,q$ lie either in the region $t\geq 0$ (as it is isometric to the closed half- plane $t\geq 0$ in $\Lo^2$ by 
\eqref{e_et_tmayor0}) or in the cylinder $t<0$, whose slices are Cauchy (apply for example \cite[Prop. 3.1]{S_Penrose}). If $t(p)<0$ and $t(q)\geq 0$, necessarily $J^-(q) \cap \{t=0\}$ is compact and, thus, lies in a compact subinterval of  $]-\pi,\pi[$ for $\theta$. Thus,   the problem  is reduced again to the case of a cylinder. 
%As the Cauchy hypersurfaces must be 1-dimensional they must be homeomorphic either to $\R$ or to $\SSS^1$ and the former is excluded by the topology of $B$.  
\end{proof}
Even though the proof of this theorem has been obtained from general simple arguments, it is illustrative to find  explicit Cauchy hypersurfaces of $B$. The  comoving  observer  $\gamma_0(t)=(t,\theta=0), t\in \R$ can be also called the {\em centered  comoving  observer} as it plays a similar role as in the open cosmological case. 
However,  the {\em singular comoving  observer} $\gamma_\pi(t)=(t,\theta=\pi), t<0$ becomes specially interesting from the causal viewpoint\footnote{\label{foot_center_univ}Notice that this observer $\gamma_\pi$ becomes clearly distinguished in our construction. Moreover,  the  invariance  under reflections (c) (below \eqref{e_varphi}) made the  comoving  observer  $\gamma_0$ distinguished too. This invariance can be regarded as the $O(n=1)$ invariance of the 2-spacetime, which will be generalized to $n\geq 2$ later.  Dropping this invariance  and the additional requirements in (c) (see footnote \ref{foot_center_univ}),  one could try to redefine $\gamma_0$  at each circle $t=t_0<0$ as the  unique point $\gamma_0(t_0)\neq \gamma_\pi(t_0)$ equidistant of $\gamma_\pi(t_0)$ from both sides in the circle.  This would permit a more intrinsic characterization of the centered \br comoving \er observer. However, in general,  the so-constructed $\gamma_0$ might not be timelike (it would be  only guaranteed that $t$ would grow along it) and our  choice  rules out this possibility.}. For each $t_0<0$, let  
$$\gamma_{\pm}^{t_0}(t)= (t, \theta^{t_0}_\pm (t)), \qquad t\geq t_0
$$ 
be the two lightlike $t$-parameterized pregeodesics starting at $\gamma_\pi(t_0)$, where $\theta^{t_0}_+ (t_0)=\pi$ and \br $\theta^{t_0}_+(t)$ \er decreases with $t$ towards $0$, \br while \er $\theta^{t_0}_- (t_0)=-\pi$ and \br $\theta^{t}_- (t)$ \er increases with $t$ towards $0$.
Notice that, once these pregeodesics abandon a small neigborhood of $\gamma_\pi(t_0)$, their coordinate $\theta$ must decrease/increase until reaching 
the value $0$ (as $\varphi$ is bounded in $\R\times [-\pi+\epsilon, \pi-\epsilon]$ for any $\epsilon>0$). Moreover, they must 
arrive at the same point in the centered  comoving \ observer $\gamma_0(t'_0)$ because of the invariance of $\varphi$ under
%\footnote{Recall also footnote \ref{foot_center_univ}.} 
$\theta
\mapsto -\theta$. Notice that  the topological circle given by these two lightlike segments is $E^+(\gamma_\pi(t_0))$, i.e., the future horismos  of $\gamma_\pi(t_0)$. This is composed by the points in the causal future  $\gamma_\pi(t_0)$  not included in the chronological one\footnote{See 
\cite[footnote 48]{S_Penrose} for additional background specific of the 2-dim. case.}. The causal future $J^+(\gamma_\pi(t_0))$ can be obtained as the union of all the  comoving  observers starting at $E^+(\gamma_\pi(t_0))$. As a consequence, $I^+(\gamma_\pi(t_0))$ includes  the region $t\geq 0$ except \br at most the \er  compact subset $K
\br  =J^-(E^+(\gamma_\pi(t_0))\cap \{t\geq 0\} \er
$ (which is included in 
$[0,\infty)\times [-\pi+\epsilon, \pi-\epsilon]$  for some $\epsilon>0$).
%, see Fig. \ref{} .

\begin{figure}
	\centering
\includegraphics[height=0.4\textheight]{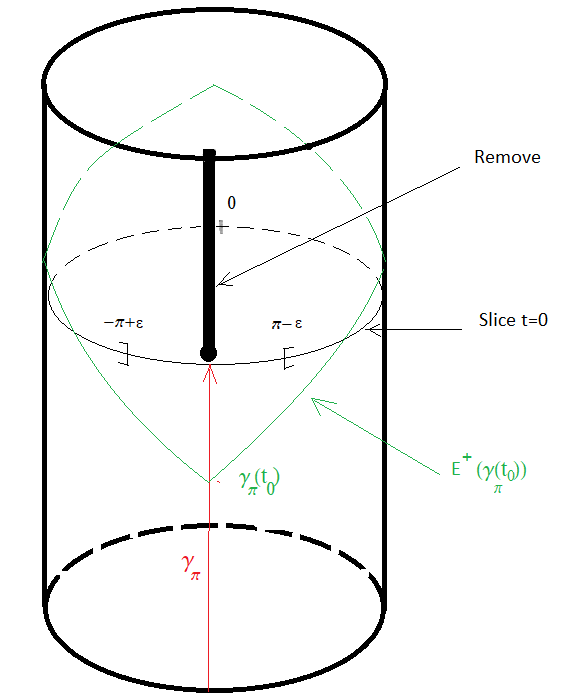}
\caption{\label{fig2} BCTCM, with the singular comoving  observer $\gamma_\pi$ in red. For $t_0<0$, $E^+(\gamma_\pi(t_0))$ (in green) is a Cauchy hypersurface homeomorphic to $\SSS^1$ which separates the spacetime in two connected open subsets. One of them is equal to  $I^+(\gamma_\pi(t_0))$. The region $\{t\geq 0\}$ is isometric to a closed half of $\Lo^2$, and $\{t\geq 0\}\setminus I^+(\gamma_\pi(t_0))$ is compact.   
}\end{figure}

Summing up, from the previous discussion (see also Fig. \ref{fig2})
\begin{prop}\label{p_sco}
In a BCTCM, the future  horismos $E^+(\gamma_\pi(t_0))$  is a Cauchy hypersurface for any point $\gamma_\pi(t_0)$, $t_0<0$, of the singular  comoving  observer. 

Moreover, the chronological future  $I^+(\gamma_\pi(t_0))$ includes  the region $t\geq 0$ except a compact subset 
and it is foliated by the Cauchy hypersurfaces 
$E^+(\gamma_\pi(t))$ with $t_0<t<0$. \end{prop}

\begin{rem}\label{r_propag0} 
By a Cauchy hypersurface $S$ we mean a subset which is crossed exactly once by any inextendible timelike curve (then,  it is necessarily a topological hypersurface but perhaps non-smooth); indeed,  $E^+(\gamma_\pi(t_0))$ lies exactly  under these  minimal hypotheses. %One can claim that, 
Once such a hypersurface is obtained, general results  ensure  the existence of an {\em acausal} one (i.e., causal curves cannot intersect $S$ in more than a point) \cite{Ge}, then   a smooth spacelike  one \cite{BS03} and, finally,   a foliation of the whole spacetime by this type of Cauchy hypersurfaces\footnote{This foliation is also endowed with a global orthogonal splitting type $(\R \times S, -\Lambda dt^2-g_t)$, which is not used here, see also the review \cite{S_Penrose} for background.} \cite{BS05}. Moreover, as the $t$-
slices $t<0$  are Cauchy  by Prop. \ref{p_dim2_a},  one can choose one  of them $\Sigma$ as the initial data for the Cauchy problem\footnote{In the case $n\geq 2$ the corresponding slices will have positive constant curvature, which can be regarded as an extra hypothesis for the Cauchy problem.}. Following  \cite{BS06}, $\Sigma$ can also be included  in a Cauchy slicing (see also \cite{Ringstrom}). 

However,   the readers can convince theirselves that all this can be done directly in the very particular case of a BCTCM; in fact, the explicit constructions above may allow one to understand better how the model works. Notice also that the slicings are consistent with Geroch's theorem \cite{Ge2} which asserts that, in any causally well behaved compact  region of a spacetime limited by two disjoint compact spacelike hypersurfaces (with no boundary) $S$ and $S'$, these two hypersurfaces, as well as any other compact spacelike one therein, must be homeomorphic. 
\end{rem}

\begin{rem} \label{r_further} 
(1) The curvature of the metric \eqref{e_closed metric} is  
$K=\frac{1}{\varphi} \frac{\partial^2 \varphi}{\partial t^2}$ (see for example \cite[Ch. 3, Prop.  44]{O}). Taking into account that, essentially, $\varphi$ is required to grow fast with $t$ for every $\theta$ is close  to   $\pm \pi$ and, eventually,  ``stabilize'' (being constant) at some $t<0$ (but not at the limit  $\theta=\pm \pi$), then $\varphi$ can be chosen so that the timelike convergence  condition (i.e., $K\leq 0$ in the case of surfaces)  holds everywhere but close to   the comoving singular observer $\gamma_\pi$ and small $t<0$.  

(2)  One could extend the definition of BCTCM by permitting that all the observers $\gamma_{\theta_0}(t)=(t,\theta_0), t<0$ are singular (in the sense of inextensible to $t=0$) for $\theta_0$ in an interval around $e^{i\pi}$ (i.e., $\theta_0\in (\pi-\delta_0, \pi] \cup [-\pi, -\pi+\delta_0)$, for some $\delta_0>0$).  In principle, this case would be straightforward from the studied one, anyway, other  weakenings  of the hypotheses on $\varphi$ are possible  (recall Remark \ref{r_marc})  and  might deserve a further study.
\end{rem}

\subsection{Spatial dimension $n\geq 2$}\label{s42} This case will be a direct extension of the previous one by using spherical coordinates  $s\in (0,\pi), \gs$ in $\SSS^n$ as in \eqref{e_spherical} (notice that, here, we start using $s$ instead of $r$ for the unit sphere $\SSS^n$). However, now the topological change will imply a curvature change.

Recall that, for $n=1$, $s=|\theta|$ plays the role of the radial spherical coordinate and the pairs $\pm\theta$ play the role of  a sphere, which is identifiable to $\gs$ when $s=1$ and collapses to a single point when $s=0,\pi$. We will use either $e^{i\pi}$ or  $s=\pi$ to denote the antipodal point of $s=0$ for our choice of spherical coordinates in $\SSS^n$.

 \begin{defi}\label{d_BCTCCM} A {\em basic cosmological  topological and  curvature change model (BCTCCM)}  is the manifold
 $$
 \R\times \SSS^n \setminus \{(t, e^{i\pi}): t\geq 0\}
 $$ 
 endowed with the metric 
\be \label{e_metrica_gen_curv_change}
g=-dt^2+ \varphi^2(t,s) ds^2+ S_{k(t)}^2(r(t,s)) g_{\SSS^{n-1}} 
\qquad (t,s)\in \R\times ]0,\pi[
\ee
  extended  naturally   to $s=0$, as well as to $e^{i\pi}$ when $t<0$, where:
  \begin{itemize}
  \item   $\varphi$ satisfies the hypotheses (a), (b), (c) below formula \eqref{e_closed metric}  (including  \br the part of (c) mentioned \er in footnote \ref{foot_extremoPi}), 
\item  $r$ is defined on $\R\times [0,\pi]$ as:
  $$
  r(t,s)= \int_0^s\varphi(t,\bar s)d\bar s
  $$
  \item we define
 $k(t)= \pi^2  /r(t,\pi)^2$, in particular, $k(t)=0$ for $t\geq 0$, in agreement with \eqref{e_curvat_top_change}.
  \end{itemize}
 \end{defi}

\begin{thm}\label{t_closed} Any BCTCCM is a smooth spacetime satisfying:

\begin{enumerate}
\item All the slices $t=t_0$ have constant curvature isometric to 
the  sphere of extrinsic radius $r(t_0,\pi)/\pi$, 
if $t_0<0$ 
and to
$\R^n$ otherwise.

\item It is globally hyperbolic, with Cauchy hypersurfaces homeomorphic to $\SSS^n$. In particular, the slices $t=t_0<0$ are Cauchy.
\end{enumerate}
\end{thm}

\begin{proof} (1) By construction,  $\varphi(t,s)$ was smooth; then, so is $r(t,s)$, as well as $k(t)$, the latter  when $t\neq 0$. To check  smoothness at $t=0$, notice that,  from \eqref{e_varphi_a}, 
$\sqrt{k(t)}=2\pi / \ell(t)$  and, whenever $1/|t|\leq m$ then $ \ell(t)\geq  2 e^m$. Thus,
$$
\lim_{t\nearrow 0}\frac{k(t)}{|t|^s}=
(2\pi)^2 \lim_{t\nearrow 0}\frac{1}{|t|^s \ell(t)^2}
\leq \pi^2 \lim_{m\nearrow \infty}\frac{m ^s}{e^{2m}}=0, 
%\qquad for s \in \N.
$$
for all $s \in \N$, which implies that the all $s$-th derivatives \ of $k(t)$ vanish.  %(use inductively L'Hopital rule); 
%(the smoothness of $\sqrt{k(t)}$ holds analogously). 
%
To analize $g$, consider first the function  $\varphi_0$ in \eqref{e_varphi_b} and change the coordinates  $(t,s)$ by $(t,r)$ in a neighborhood of the region $t\geq 0$ so that $dr=\varphi_0(s)ds$ and  
\begin{equation}
\label{e_euclidea_varphi00} g=-dt^2+ \varphi_{0}^2(s)ds^2+ S_{k(t)}^2(r(t,s))\gs= -dt^2+ dr^2+ S_{k(t)}^2(r)\gs .  \end{equation} 
 This  is  a smooth metric as in Theorem~\ref{t_local} and, so, the  smoothness of $g$  in the whole region
$(t,s)\in \R\times ]0,\pi[$  is straightforward. The required constant curvature of  the $t_0$-slices follows from \eqref{e_euclidea_varphi00} when $t\geq 0$ and, otherwise, by  changing  the coordinate $s$ by $r(t_0,s)$ at each slice with $t_0< 0$. 
 %follows applying the local case  in Theorem~\ref{t_local} to \eqref{e_euclidea_varphi00}.
%where no restriction on the sing of the curvature was done. 
For the smoothness of the extension to $s=0$, the proof of Theorem~\ref{t_open} works with no modification as  $\varphi$ is constantly equal to 1 around $s=0$ (\br even though this could be relaxed, \er recall footnote~\ref{foot_extremoPi}). 

Next, let us check the smoothness of  the extension of $g$ to  $s=\pi$ when $t<0$. Essentially, this will also be  reduced to the proof of the  case $s=0$ in Theorem~\ref{t_open}. For this purpose,  choose $t_1<0$ and  let $\epsilon_0>0$ such that $\varphi(t,s)$ is independent of $s$  around $(t_1,\pi)$, that is, 
$$
 \varphi_\pi(t):= \varphi(t,\pi)   = \varphi (t,s), \qquad 
\forall t\in [t_1-\epsilon_0, t_1+\epsilon_0], \quad \forall s\in [s_0:=\pi-\epsilon_0,\pi].$$ 
% (recall again footnote \ref{foot_extremoPi}).
 Then, in this region define:
%\footnote{se podri'a intentar (compa'rese con el ape'ndice): regarding $r$ as a coordinate instead of $s$ one has:
%$$
%g+dt^2=\left(dr^2+S_{k(t)}^2(r) g_{\SSS^{n-1}}\right)+(\partial _t\varphi)^2(t,s(r,t)) \, dt^2- \partial _t\varphi(t,s(r,t)) 2 \,dt \, dr
%$$
%Notice that, for each $t<0$,  the function $\varphi$ is constant in $s$ (recall again footnote \ref{foot_extremoPi}), thus, the function $s(r,t)$ is smooth therein. Sin embargo, esto tiene dos problemas. El primero es que el te'rmino $dtdr$ no esta' bien definido (un vector $v$ y $-v$ tendrian igual componente $dr$ en el limite) y el segundo que la variacion de $r_{k(t)}$ debe tenerse en cuenta en el pare'ntesis grande. Por lo que solo quedari'a rezar para que ambos problemas se compensen... }
\begin{equation}\label{e_rts}
 r(t,s) =\int_0^s\varphi(t,\bar s)d\bar s= 
r(t,\pi)-  \varphi_{\pi}(t)  (\pi-s), \qquad \hbox{when} \;  s_0 < s \leq \pi . 
%\quad t_0-\epsilon_0\leq t \leq t_0+\epsilon_0.   
\end{equation}
Let us introduce the functions $\bar s= \pi-s$ and $\bar r(t,s)=r(t,\pi)-r(t,s)$.  We have just to prove the smoothness  $g$ (given by  \eqref{e_metrica_gen_curv_change}) in coordinates $(t,\bar r)$ at $\bar{r}=0$ (i.e., $\bar s=0$). Using \eqref{e_rts}, 
$$
 \bar r(t,s)= \varphi_{\pi}(t)(\pi-s)= \varphi_{\pi}(t)\bar s, \qquad d \bar r = \varphi_{\pi}(t) d \bar s + \dot \varphi_{\pi}(t) \bar s dt,  
$$
 where $\dot \varphi_{\pi}$ denotes derivative, and we have 
\begin{equation}\label{e_cross}
\varphi(t,s)^2 d  s^2=\varphi_{\pi}(t)^2 d \bar s^2 =
d\bar r^2 + \frac{\dot \varphi_{\pi}(t)^2}{\varphi_{\pi}(t)^2} \bar r^2 dt^2 -\frac{\dot \varphi_{\pi}(t)}{\varphi_{\pi}(t)} 2 \bar r d\bar r dt.
\end{equation}
As $\bar s^2$ is a smooth function in our region, so is $\bar r^2$ as well as $\bar r d\bar r= d(\bar r^2/2)$. 
\br Substituting \eqref{e_cross} in \eqref{e_metrica_gen_curv_change}, 
the last two terms in \eqref{e_cross} become irrelevant for  the smoothness of $g$. So, the problem reduces to check the smoothness at $\bar r=0$ of the terms  
\begin{equation}\label{e_Marc}
d\bar r^2 +S_{k(t)}(r(t,s))\gs  
\end{equation}
where the function $S_{k(t)}(r(t,s))$ must be expressed in the coordinates $(t,\bar r)$. Using the expressions of $S_k$,  $k(t)$ for a BCTCCM and  $\bar r(t,s)$ above: 
\begin{equation}\label{e_Sbarra_r}
\begin{array}{rl}  S_{k(t)}(r(t,s))= &  \frac{r(t, \pi)}{\pi}  \sin \left(\pi \, \frac{r(t,s)}{r(t,\pi)} \right)
 =  \frac{r(t, \pi)}{\pi}  \sin \left(\pi \, \left(1-\frac{\bar r(t,s)}{r(t,\pi)}\right) \right)
 \\= &
\frac{r(t, \pi)}{\pi}  \sin \left(\pi \, \frac{\bar r(t,s)}{r(t,\pi)} \right)  = 
S_{k(t)}(\bar r(t,s)).
\end{array}
\end{equation}
That is, $S_{k(t)}(r(t,s))$ becomes $S_{k(t)}(\bar r)$
and the smoothness of \eqref{e_Marc} follows from  Theorem~\ref{t_open}. \er

(2) All the arguments in Prop. \ref{p_dim2_a} can be applied for this part. In particular, applying \eqref{e_euclidea_varphi00}, \eqref{e_curvat_top_change},
 the metric in 
  the region  $t\geq 0$ becomes 
\begin{equation}
\label{e_euclidea_varphi0} g=-dt^2+ \varphi^2_0(s)ds^2+ r^2(s)\gs= -dt^2+ dr^2+ r^2\gs 
\end{equation} 
with $r(s)=\int_0^s\varphi_0(\bar s) d\bar s \in ]0,\infty[$, which is isometric to the standard half space $t\geq 0$ in $\LL^{n+1}$.
So, for any  point $q$ with $t(q)\geq 0$, necessarily $J^-(q)\cap \{t=0\}$ is compact. Then,   the slices $t=t_0<0$ (which are homemorphic to a sphere) are Cauchy, and  those with $t_0\geq 0$ (homemorphic to $\R^n$) cannot.    
\end{proof}

%\begin{rem} \label{r_propag}  Some of the previous constructions could be circumvented. Indeed, as commented in  Rem. \ref{r_propag0}, once  the spacetime is known to be globally hyperbolic with one  Cauchy hypersurface homeomorphic to a sphere, the existence of a spacelike Cauchy slicing by spheres  follows 
%directly. Moreover, as Theorem \ref{t_closed} ensures  that the %$t$-
%slices $t<0$  are Cauchy (and have constant positive curvature),   one can choose one  of them $\Sigma$ as the initial data for the Cauchy problem. Following  \cite{BS06}, $\Sigma$ can also be included  in a Cauchy slicing (see also \cite{Ringstrom}). However, the explicit constructions here may allow one to understand better how the model works. 
%\end{rem}

\begin{rem}\label{r_new} 
Recall that the BCTCCM metric \eqref{e_metrica_gen_curv_change} is more general than the original one in Theorem \ref{t_local},  because if the coordinate $s$ is replaced by  $r(t,s)$ on the whole spacetime then additional cross terms (as in \eqref{e_cross}) may appear. However,  this can be done just around the comoving singular observer, Rem. \ref{r_further}.

\end{rem}

\subsection{Further transitions}\label{s43}
Once the BCTCCM has been constructed, one can  combine it directly with the open model to obtain a smooth transition from positive to negative curvature, namely: (i) choose a smooth function $k(t)$ with, say,   %\br $k(0)=0$,  
$k(t)>0$, if $t<0$, and $k(t)<0$ \br if $t>0$, \er
%(one can simplify checking technicalities assuming that all the derivatives of $k(t)$ vanish at 0), 
(ii) in the region $t\leq 0$ choose the metric of the BCTCCM in Defn. \ref{d_BCTCCM},    
(iii) in the region $t>0$ choose the metric of the open cosmological model  in Theorem \ref{t_open}  with the following caution:  rewrite this metric  by using the coordinate $s$ and the function $\varphi_0$ in \eqref{e_euclidea_varphi0} (so that it will  match smoothly with the BCTCCM at $t= 0$). 
It is worth pointing out, about this procedure:
\ben 
\item The change of coordinates in  the step (iii) rellabels the  comoving observers, but these observers are not modified by such a change (simply, the coordinates $r$ is changed into $s$, and the coordinate $t$ is not involved in such a change).

\item Such a model with strict transition from $k(t)<0$ to $k(t)>0$ is again globally hyperbolic with {\em compact} Cauchy hypersurfaces. Indeed, one can reason this as in the proof of Theorem \ref{t_closed} and Prop. \ref{p_dim2_a}. 

However, the following more straightforward reasoning holds. As argued in the proof of Theorem \ref{t_open},  the cones of the constructed model in $t>0$  are narrower than those of  $\Lo^{n+1}$ spacetime (in the chosen coordinates) and, morever, $\Lo^{n+1}$ is the metric of the original BCTCCM in this region.  Thus, the Cauchy hypersurfaces and foliations obtained for the BCTCCM in \S \ref{s42} remain Cauchy for the strict transition here. 

\item The fact that each slice $\{t=t_0\}$   was a Cauchy hypersurface for $t_0<0$ but it  is only a partial Cauchy one for $t_0>0$ becomes  geometrically evident. 
However, this  might not be evident for the  comoving  observers. In fact, these slices are associated with  their restspaces. The  infinitesimal  and, eventually, local measures of these spaces may be achieved as a consequence of the principle of equivalence, but it is not straightforward how to make  global measurements. Recall that the
``disappearance'' of the  singular comoving  observer cannot be seen directly by any (comoving  or not) observer, as the 
 spacetime is globally hyperbolic and, thus, free of naked singularities.
\een

As commented in the Introduction and Rem. \ref{r1}, the possibility of  the existence of a spacetime
in which the topology of
certain geometrically preferred sections is changing in time have been considered some times in the literature. The starting point was  Stephani's study of  the class of spacetimes embeddable in a flat five-dimensional space \cite{St} (see also \cite[Chapter 15]{St_et_al}), which included the now so-called Stephani Universes \cite{Kr3}. 
Especially, Krasi\'nski and later Sussman  
%considered  the Stephani class of universes  and 
developed both the local  \cite{Kr1, Su0} and the global viewpoints   \cite{Kr2, Su} (see also the short overview in \cite[p. 675]{Kr2} and the book \cite{Kr3}).  However, their viewpoint is different to ours. 

In the case of Krasi\'nski \cite{Kr1, Kr2}, de Sitter 4-spacetime serves as the qualitative model  for the curvature sign changing foliations (see %this is emphasized in 
\cite{Kr2}). The leaves of the foliation are not the slices of the original universal time $t$ in \eqref{e_metrica_gen} (which is generically a priviledged time function, in a similar way as the FLRW case); thus, the leaves are not orthogonal to the comoving observers at $\partial_t$. Moreover,   the coordinate $r$ is considered globally and, as pointed out in Remark \ref{r1} (2)  (see also the Appendix), this may introduce smoothability issues. 
Sussman \cite{Su0, Su} considered an  expression of the 4-metric which  separates the cases of   positive, negative and zero spatial curvature (see  formulas (1) and  (2) in these references) and he studied systematically the cases  of  one, two or zero \br comoving \er centers permitted by $O(3)$ symmetry.
%, the first two ones consistent  with our open and closed models,  respectively. Even though Sussman general methods may be applicable to our case, 
 In comparison,  our direct approach gives a  straight geometric picture where both global hyperbolicity and an  explicit Cauchy slicing emerge naturally. %(compare with \cite[\S III C]{Su}). 
  
Other topological transitions between spatially
compact and non-compact universes as those in \cite[\S 5.2]{KM} are quite different to ours. %nor the negative curvature slices intersect the comoving observers too far from the centered one. 

%When considering a cosmological model, however, one must consider  the {\em global} possibility of a sign change for the  foliation given by the {\em prescribed} universal time $t$ (which use to be priviledged, as generically occurs in the FLRW case). This global sign change is  subtler because, on the one hand, it  may  imply a  change of the topology of the $t$-slices and, on the other, it introduces further issues of smoothability, which might lead to unphysical restrictions on the stress-energy tensor.   
%The  local change of sign appears , when studying the class of spacetimesembeddable in a flat five-dimensional space.  Krasi\'nski \cite{Kr1} rediscovered these examples by assuming  $O(3)$ symmetry on the spatial slices (see also \cite[\S 15.3]{St_et_al}). He also considered some global possibilities \cite{Kr2}; however, some non-trivial issues  (to be discussed  below) appear.  
 %, even in the cases that the assumptions on spatial isotropy permit the sign change, see Remark \ref{r1}. 
%However, Stephani \cite{St}  

%\begin{rem} \label{rUF}

\section{Conclusions}

We have carried out a direct study of a class of simple cosmological models  (related to a more general class of spacetimes studied  in dimension 4 by Stephani  \cite{St, St_et_al}) %, Krasi\'nski \cite{Kr1, Kr2, Kr3} and Sussman \cite{Su0, Su}) 
with  a universal time function $t$ giving  rise to freely following \br comoving \er observers  whose restspaces $\{t=t_0\}$ have  constant  curvature $k(t_0)$ and vary with $t_0$, this variation including its sign. %Such a possibility seems to have been dismissed so far because of an incomplete   geometric development of  cosmological  hypotheses about spatial isotropy and homogeneity. 
We have focused on  a rigurous mathematical presentation of the models, which permit a direct comparison with usual hypotheses on isotropy and homogeneity.  From the local viewpoint, this includes a detailed study of the regularity of the metric. From the global one, our models are globally hyperbolic spacetimes and we have found  two natural classes. The {\em open models} ($k(t)\leq 0$ everywhere) have a simple intuitive global  structure.  They    distinguish a {\em centered   comoving  observer} which, in certain sense, is the center of the spatial expansion or contration governed by $k(t)$. The {\em closed models} ($k(t)> 0$ somewhere) are much subtler globally. This happens because the $t$-slices must present a topology change (from $\SSS^n$ when $k(t)>0$ to $\R^n$), which has to be compatible with the rigid product topological structure of  any globally hyperbolic spacetime  (that is, the topology $\R\times \Sigma$, where  $\Sigma$ is any Cauchy hypersurface). In a natural way, this leads to the existence of  a   {\em singular  comoving  observer  $\gamma_\pi$} (in addition to the previous  centered one). From the global viewpoint, this is the truly priviledged observer, as it makes apparent the Cauchy splitting; in fact,  the centered observer can be regarded only as the ``farthest'' one at each $t$-slice (see footnote \ref{foot_center_univ}). 

 These models  release some possibilities which might fit in current  Cosmology. For example, the open models are consistent with a flat space at some ``universal instant'' $\{t=t_0\}$ and a bouncing therein (Ex. \ref{e_sinh}, footnote \ref{f_sinh}). The  closed ones might apply to  inflationary processes along the singular  comoving  observer $\gamma_\pi$,  
 %This observer may determine the centered comoving observer  in a more intrinsic way (see footnote \ref{foot_center_univ}). 
% Thus,   $\gamma_\pi$  is 
which become the key for both, the spatial topology change and the  center of the expansion.  In fact,  $\gamma_\pi$  lies in the region  $t<0$ of the time function $t$ associated with the constant curvature slices, but the spacetime has no naked singularites. So, the singular \br comoving \er observer will  ``live'' for every instant of any Cauchy time function, as any other (inextendible) observer. 
 These features might attract the attention of the community and give rise to  observational issues (recall \cite{Ellis1,Ellis2})  to be studied further.

\section*{Acknowledgments}
 The author warmly acknowledges useful and encouraging discussions with Rodrigo \'Avalos (UF. do Cear\'a, Fortaleza) and   the  reading, comments  and suggested references by  Marc Mars (U. Salamanca), Jos\'e M.M. Senovilla and Ra\"ul Vera  (both at UPV-EHU, Bilbao), \br as well as the careful reading and suggestions by the referee. \er  
  Partially supported by
the grants A-FQM-494-UGR18 (Junta de Andaluc\'ia/FEDER), PID2020-116126GB-I00 (MCIN/ AEI/10.13039/501100011033), 
%P20-01391 (PAIDI 2020, Junta de Andaluc\'{\i}a) 
 and the framework IMAG/  Mar\'{\i}a de Maeztu,   CEX2020-001105-MCIN/ AEI/ 10.13039/501100011033.
 
% \bibliographystyle{abbrv}
%\bibliography{mybib.bib}
%\end{document}

\section*{Data Availability and conflict of interest statements}
Data sharing is not applicable to this article as no datasets were generated or analysed during the current study.

The author states that there is no conflict of interest.

\section*{\br Appendix:  expression of the Ricci %and Einstein 
tensor \er }

\br
The computation of the Ricci tensor   of the fundamental metric in Theorem  \ref{t_local}  can be carried out by using its warped structure as in \cite[Corollary 7.43]{O}. Indeed, for $X,Y \in $ Span$\{\partial_t, \partial_r\}$ one has:
$$
\hbox{Ric}(X,Y)=-\frac{n}{S_{k(t)}(r)} \hbox{Hess}(S(t,r))(X,Y)
$$ 
with $S(t,r)$ as in \eqref{e1} with Hessian expressable in terms of the one of $S_k(r)$:
$$
\begin{array}{rlll}
\partial_r S_k(r)= C_k(r) & & &
\partial_k S_k(r)= \frac{1}{2k}(rC_k(r)-S_k(r)) 
\\
\partial_r^2 S_k(r)= -k S_k(r) & & &
\partial_{rk}^2 S_k(r)= -\frac{r}{2} S_k(r)
\end{array}
$$
$$
\begin{array}{l}
\partial_k^2 S_k(r)= \frac{1}{4k^2}\left(-3rC_k(r)+(3-kr^2)S_k(r)\right), 
\end{array}
$$ 
(notice $\partial_k C_k(r)=-rS_k(r)/2$, 
$\lim_{k\rightarrow 0}\partial_k S_k(r)=-r^3/6$, 
$\lim_{k\rightarrow 0}\partial^2_k S_k(r)=r^5/60$). Thus, $\partial_t S(t,r)=k'(t) \partial_k S_{k(t)}(r)$ and:
$$
\begin{array}{llll}
\partial_r^2 S(t,r)= -k(t) S_{k(t)}(r) & & &
\partial_{rt}^2 S(t,r)= -\frac{k'(t) r}{2} S_{k(t)}(r)
\end{array}$$
$$
\begin{array}{l}
\partial_{tt}^2 S(t,r)= \frac{2k(t)k''(t)-3k'(t)^2}{4 k(t)^2}(rC_{k(t)}(r)-S_{k(t)}(r))-  \frac{k'(t)^2}{4k(t)} r^2 S_{k(t)}(r) 
\end{array}
$$
So, the Ricci tensor for the base of the warped product is:
$$
\begin{array}{llll}
\hbox{Ric}(\partial_r,\partial_r)=nk(t)
 & & &
\hbox{Ric}(\partial_r,\partial_t)= \frac{n}{2}k'(t) r
\end{array}$$
$$
\begin{array}{l}
\hbox{Ric}(\partial_t,\partial_t)
= -n \frac{2k(t)k''(t)-3k'(t)^2}{4 k(t)^2}(\frac{r}{T_{k(t)}(r)}-1)+n 
  \frac{k'(t)^2}{4k(t)} r^2 
\end{array}
$$
Now, if $V,W$ are vectors tangent to the fiber $(S^{n-1},\gs)$ one has Ric$(X,V)=0$ (as in any warped product) and 
$$
\hbox{Ric}(V,W)=[(n-1)-S \Delta S+(n-2)dS(\hbox{grad}(S))] \; \gs(V,W)
$$
where $\Delta S$, $\hbox{grad}(S)$ denote, resp., the Laplacian and gradient of $S(t,r)$ for the metric $-dt^2+dr^2$, thus:
$$
\begin{array}{rl}
dS(\hbox{grad}(S))= & 
-\frac{k'(t)^2}{4k^2}(rC_k(r)-S_k(r))^2+C_k(r)^2
\\
\Delta S= &  -\frac{2k(t)k''(t)-3k'(t)^2}{4 k(t)^2}(rC_{k(t)}(r)-S_{k(t)}(r))+  \frac{k'(t)^2 r^2 -4k(t)^2}{4k(t)}  S_{k(t)}(r).
\end{array}
$$
\er

\end{document}